\documentclass[12pt,a4paper,oneside,reqno]{amsart}

\usepackage{amssymb,amsmath}
\usepackage{stackrel}
\usepackage{color}
\usepackage{tikz}

\definecolor{thmcolor}{rgb}{0,0,.4} 
\definecolor{remarkcolor}{rgb}{0,.1,0} 
\definecolor{quecolor}{rgb}{.2,.2,0} 
\definecolor{axcolor}{rgb}{.3,0,.3}
\definecolor{thmbgcolor}{rgb}{0.9,0.9,1} 
\definecolor{rmbgcolor}{rgb}{0.6,1,0.6} 
 
\theoremstyle{definition} \newtheorem{thm}{\colorbox{thmbgcolor}{\textcolor{thmcolor}{Theorem}}}[section] 
\theoremstyle{definition} \newtheorem{cor}[thm]{\colorbox{thmbgcolor}{\textcolor{thmcolor}{Corollary}}} 
\theoremstyle{definition} \newtheorem{lem}[thm]{\colorbox{thmbgcolor}{\textcolor{thmcolor}{Lemma}}}
\theoremstyle{definition} 
 
\theoremstyle{remark}  
\theoremstyle{remark} \newtheorem{conj}[thm]{\colorbox{rmbgcolor}{\sc\textcolor{remarkcolor}{Conjecture}}}
\theoremstyle{definition}  
\theoremstyle{definition} \newtheorem{rem}[thm]{\colorbox{rmbgcolor}{\textcolor{remarkcolor}{Remark}}}
\theoremstyle{definition}

\newcommand{\taxis}{\mathsf{t\text{-}axis}}
\newcommand{\timed}{\mathsf{time}}
\newcommand{\sqspace}{\mathsf{space}^2}
\newcommand{\vo}{\mathbf{\bar o}}
\newcommand{\vx}{\mathbf{\bar x}}
\newcommand{\vy}{\mathbf{\bar y}}
\newcommand{\vv}{\mathbf{\bar v}}

\newcommand{\vw}{\mathbf{\bar w}}
\newcommand{\vet}{\mathbf{\bar1}}
\newcommand{\de}{:=}
\newcommand{\rac}{\mathbb{Q}}
\newcommand{\R}{\mathbb{R}}

\newcommand{\defiff}{\ \stackrel{\;def}{\Longleftrightarrow}\ }
\newcommand{\IOb}{\ensuremath{\mathsf{IOb}}} 
\newcommand{\B}{\ensuremath{\mathit{B}}} 
\newcommand{\Ph}{\ensuremath{\mathsf{Ph}}} 
\newcommand{\Q}{\ensuremath{\mathit{Q}}} 
\newcommand{\W}{\ensuremath{\mathsf{W}}} 
\newcommand{\ev}{\ensuremath{\mathsf{ev}}} 
\newcommand{\Id}{\ensuremath{\mathsf{Id}}} 
\newcommand{\wl}{\ensuremath{\mathsf{wl}}}
\newcommand{\w}{\ensuremath{\mathsf{w}}}
\newcommand{\num}{\textit{Num}}
\newcommand{\ax}[1]{\textcolor{axcolor}{\ensuremath{\mathsf{#1}}}} 

\begin{document}
\title{Special Relativity over the Field of Rational Numbers}
\author{Judit X.~Madar\'asz and Gergely Sz\'ekely} \address{Alfr\'ed
  R\'enyi Institute of Mathematics \\ Hungarian Academy of
  Sciences \\ Re\'altanoda utca 13-15, H-1053, Budapest, Hungary}
\email{\{madarasz.judit, szekely.gergely\}@renyi.mta.hu}
\date{\today}

\begin{abstract}
We investigate the question: what structures of numbers (as physical
quantities) are suitable to be used in special relativity?  The answer to
this question depends strongly on the auxiliary assumptions we add to
the basic assumptions of special relativity. We show that there is a
natural axiom system of special relativity which can be modeled even
over the field of rational numbers.
\end{abstract}
\maketitle

\section{Introduction}

In this paper, we investigate, within an axiomatic framework, the
question: what structures of numbers (as physical quantities) are
suitable to be used in special relativity?  There are several reasons
to investigate this kind of questions in the case of any theory of
physics. First of all, we cannot experimentally verify whether the
structure of quantities is isomorphic to the field of real
numbers. Moreover, the fact that the outcome of every measurement is a
finite decimal suggests that rational numbers (or even integers)
should be enough to model physical quantities. Another reason is that
these investigations lead to a deeper understanding of the connection
of the mathematical assumptions about the quantities and the other
(physical) assumptions of the theory. Hence these investigations lead
to a deeper understanding of any theory of physics, which may  come
handy if we have to change some of the basic assumptions for some
reason. For a more general perspective of this research direction, see
\cite{wnst}.

So in general we would like to investigate the question
\begin{center}
{\it ``What structure can numbers have in a certain physical theory?''}
\end{center}
To introduce the central concept of our investigation, let \ax{Th} be
a theory of physics that contains the concept of numbers (as physical
quantities) together with some algebraic operations on them (or at
least these concepts are definable in \ax{Th}). In this case, we can
introduce notation $\num(\ax{Th})$ for the class of the possible
quantity structures of theory \ax{Th}:
\begin{multline}
\num(\ax{Th})=\{\mathfrak{Q}:\mathfrak{Q} \text{ is a structure of
  quantities}\\\text{over which \ax{Th} has a model.}\}
\end{multline}

In this paper, we investigate our question only in the case of special
relativity. However, this question can be investigated in any other
physical theory the same way.

We will see an axiom on observers implying that positive numbers have
square roots. Therefore, we recall that {\bf Euclidean fields}, which
got their names after their role in Tarski's first-order logic
axiomatization of Euclidean geometry \cite{TarskiElge}, are ordered
fields in which positive numbers have square roots.

Our axiom system for $d$-dimensional special relativity (\ax{SpecRel_d},
see p.\pageref{specrel}) captures the kinematics of special relativity
perfectly if $d\ge3$, see Theorem~\ref{thm-poi}.  Without any extra assumptions
\ax{SpecRel_d} has a model over every ordered field, i.e.,
\begin{equation*}
\num(\ax{SpecRel_d})=\{\mathfrak{Q}:\mathfrak{Q} \text{ is an ordered
  field} \},
\end{equation*}
see Remark~\ref{rem-of}. Therefore, \ax{SpecRel} has a model over
$\rac$ (the field of rational numbers), too.  However, if we assume
that inertial observes can move with arbitrary speed less than that of
light, see \ax{AxThExp} on p.\pageref{axthexp}, then every positive
number has to have a square root if $d\ge 3$ by Theorem~\ref{thm-eof},
i.e.,
\begin{equation*}
\num(\ax{SpecRel_d} + \ax{AxThExp})=\{\mathfrak{Q}:\mathfrak{Q} \text{
  is a Euclidean field} \} \text{ if } d\ge3,
\end{equation*}
see \cite{wnst}. In particular, the number structure cannot be the
field of rational numbers if \ax{AxThExp} is assumed and $d\ge3$, i.e.,
\begin{equation*}
\rac\not\in\num(\ax{SpecRel_d} + \ax{AxThExp})\text{ if } d\ge3.
\end{equation*}

Theorem~\ref{thm-rac}, the main result of this paper, shows that our
axiom system \ax{SpecRel} has a model over $\rac$ (in any dimension)
if we assume axiom \ax{AxThExp} only approximately, i.e.,
\begin{equation*}
\rac\in\num(\ax{SpecRel_d} + \ax{AxThExp^-})\text{ if } d\ge2,
\end{equation*}
 see the precise formulation of \ax{AxThExp^-} on
 p.\pageref{axthexp-}. Assuming \ax{AxThExp^-} instead of \ax{AxThExp}
 is reasonable because we cannot be sure in anything perfectly
 accurately in physics. Theorem~\ref{thm-rac} implies that
 \ax{SpecRel} + \ax{AxThExp} can be modeled over every subfield of the
 field of real numbers ($\R$), see Corollary~\ref{cor-arch}; and we
 conjecture that this axiom system has a model over every ordered
 field, see Conjecture~\ref{conj-of}.

An interesting and related approach of Mike Stannett introduces two
structures one for the measurable numbers and one for the theoretical
numbers and assumes that the set of measurable numbers is dense in the
set of theoretical numbers, see \cite{Mike}.

We chose first-order predicate logic to formulate our axioms because
experience (e.g., in geometry and set theory) shows that this logic is
the best logic for providing an axiomatic foundation for a theory. A
further reason for choosing first-order logic is that it is a well
defined fragment of natural language with an unambiguous syntax and
semantics, which do not depend on set theory. For further reasons,
see, e.g., \cite[\S Why FOL?]{pezsgo}, \cite{ax}, \cite[\S 11]{SzPhd},
\cite{vaananen}, \cite{wolenski}.

\section{The language of our theories}
\label{lang-s}

To our investigation, we need an axiomatic theory of special
relativity. Therefore, we will recall our axiom system \ax{SpecRel_d} in Section~\ref{ax-s}.  To
write up any axiom system, we have to choose the set of basic symbols
of its language, i.e., what objects and relations between them
will be used as basic concepts.

Here we will use the following two-sorted\footnote{That our theory is
  two-sorted means only that there are two types of basic objects
  (bodies and quantities) as opposed to, e.g., Zermelo--Fraenkel set
  theory where there is only one type of basic objects (sets).}
language of first-order logic (FOL) parametrized by a natural number
$d\ge 2$ representing the dimension of spacetime:
\begin{equation}
\{\B,\Q\,;  \IOb, \Ph,+,\cdot,\le,\W\},
\end{equation}
where $\B$ (bodies) and $\Q$ (quantities) are the two sorts,
 $\IOb$ (inertial observers) and $\Ph$ (light
signals) are one-place relation symbols of
sort $\B$, $+$ and $\cdot$ are two-place function symbols of sort
$\Q$, $\le$ is a two-place relation symbol of sort $\Q$, and $\W$ (the
worldview relation) is a $d+2$-place relation symbol the first two
arguments of which are of sort $\B$ and the rest are of sort $\Q$.

Relations $\IOb(m)$ and $\Ph(p)$ are translated as ``\textit{$m$ is an
  inertial observer},'' and ``\textit{$p$ is a light signal},''
respectively. To speak about coordinatization of observers, we
translate relation $\W(k,b,x_1,x_2,\ldots,x_d)$ as ``\textit{body $k$
  coordinatizes body $b$ at space-time location $\langle x_1,
  x_2,\ldots,x_d\rangle$},'' (i.e., at space location $\langle
x_2,\ldots,x_d\rangle$ and instant $x_1$).

{\bf Quantity  terms} are the variables of sort $\Q$ and what can be
built from them by using the two-place operations $+$ and $\cdot$,
{\bf body terms} are only the variables of sort $\B$.
$\IOb(m)$, $\Ph(p)$, $\W(m,b,x_1,\ldots,x_d)$, $x=y$, and $x\le y$
where $m$, $p$, $b$, $x$, $y$, $x_1$, \ldots, $x_d$ are arbitrary
terms of the respective sorts are so-called {\bf atomic formulas} of
our first-order logic language. The {\bf formulas} are built up from
these atomic formulas by using the logical connectives \textit{not}
($\lnot$), \textit{and} ($\land$), \textit{or} ($\lor$),
\textit{implies} ($\rightarrow$), \textit{if-and-only-if}
($\leftrightarrow$) and the quantifiers \textit{exists} ($\exists$)
and \textit{for all} ($\forall$).

To make them easier to read, we omit the outermost universal
quantifiers from the formalizations of our axioms, i.e., all the free
variables are universally quantified.

We use the notation $\Q^n$ for the set of all $n$-tuples of elements
of $\Q$. If $\vx\in \Q^n$, we assume that $\vx=\langle
x_1,\ldots,x_n\rangle$, i.e., $x_i$ denotes the
$i$-th component of the $n$-tuple $\vx$. Specially, we write $\W(m,b,\vx)$ in
place of $\W(m,b,x_1,\dots,x_d)$, and we write $\forall \vx$ in place
of $\forall x_1\dots\forall x_d$, etc.

We use first-order logic set theory as a meta theory to speak about model
theoretical terms, such as models, validity, etc.  The {\bf models} of
this language are of the form
\begin{equation}
{\mathfrak{M}} = \langle \B, \Q; \IOb_\mathfrak{M},\Ph_\mathfrak{M},+_\mathfrak{M},\cdot_\mathfrak{M},\le_\mathfrak{M},\W_\mathfrak{M}\rangle,
\end{equation}
where $\B$ and $\Q$ are nonempty sets, $\IOb_\mathfrak{M}$ and
$\Ph_\mathfrak{M}$ are subsets of $\B$, $+_\mathfrak{M}$ and
$\cdot_\mathfrak{M}$ are binary functions and $\le_\mathfrak{M}$ is a
binary relation on $\Q$, and $\W_\mathfrak{M}$ is a subset of
$\B\times \B\times \Q^d$.  Formulas are interpreted in $\mathfrak{M}$
in the usual way.  For the precise definition of the syntax and
semantics of first-order logic, see, e.g., \cite[\S 1.3]{CK}, \cite[\S
  2.1, \S 2.2]{End}.

\section{Axioms for special relativity}
\label{ax-s}
Now having our language fixed, we can recall axiom system
\ax{SpecRel_d}, as well as two theorems on \ax{SpecRel_d} related to our
investigation.

The key axiom of special relativity states that the speed of light is
the same in every direction for every inertial observers.
\begin{description}
\item[\underline{\ax{AxPh}}] For any inertial observer, the speed of
  light is the same everywhere and in every direction (and it is
  finite). Furthermore, it is possible to send out a light signal in
  any direction (existing according to the coordinate system)
  everywhere:
\begin{multline*}
\IOb(m)\rightarrow \exists c_m\Big[c_m>0\land \forall \vx\vy
 \Big(\exists p \big[\Ph(p)\land \W(m,p,\vx)\\\land
    \W(m,p,\vy)\big] \leftrightarrow \sqspace(\vx,\vy)= c_m^2\cdot\timed(\vx,\vy)^2\Big)\Big],\footnotemark
\end{multline*}
\end{description}
where $\sqspace(\vx,\vy)\de(x_2-y_2)^2+\ldots +(x_d-y_d)^2$ and
$\timed(\vy,\vy)\de x_1-y_1$. \footnotetext{That is, if $m$ is an
  inertial observer, there is a is a positive quantity $c_m$ such that
  for all coordinate points $\vx$ and $\vy$ there is a light signal
  $p$ coordinatized at $\vx$ and $\vy$ by observer $m$ if and only if
  equation $\sqspace(\vx,\vy)= c_m^2\cdot\timed(\vx,\vy)^2$ holds.}

To get back the intended meaning of axiom \ax{AxPh} (or even to be
able to define subtraction from addition), we have to assume some
properties of numbers. 

In our next axiom, we state some basic properties of addition,
multiplication and ordering true for real numbers.
\begin{description}
\item[\underline{\ax{AxOField}}]
 The quantity part $\langle \Q,+,\cdot,\le \rangle$ is an ordered field, i.e.,
\begin{itemize}
\item  $\langle\Q,+,\cdot\rangle$ is a field in the sense of abstract
algebra; and
\item 
the relation $\le$ is a linear ordering on $\Q$ such that  
\begin{itemize}
\item[i)] $x \le y\rightarrow x + z \le y + z$ and 
\item[ii)] $0 \le x \land 0 \le y\rightarrow 0 \le xy$
holds.
\end{itemize}
\end{itemize}
\end{description}

Using axiom \ax{AxOFiled} instead of assuming that the structure of
quantities is the field of real numbers not just makes our theory more
flexible, but also makes it possible to meaningfully investigate our
main question.  Another reason for using \ax{AxOField} instead of $\R$
is that we cannot experimentally verify whether the structure of
physical quantities are isomorphic to $\R$. Hence the assumption that
the structure of quantities is $\R$ cannot be empirically
supported. The two properties of real numbers which are the most
difficult to defend from empirical point of view are the Archimedean
property, see \cite{Rosinger08}, \cite[\S
  3.1]{Rosinger09},\cite{Rosinger11b}, \cite{Rosinger11a}, and the
supremum property.\footnote{The supremum property (i.e., that every
  nonempty and bounded subset of the numbers has a least upper bound)
  implies the Archimedean property. So if we want to get ourselves
  free from the Archimedean property, we have to leave this one, too.}

We also have to support \ax{AxPh} with the assumption that all
observers coordinatize the same ``external" reality (the same set of
events).  By the {\bf event} occurring for observer $m$ at point
$\vx$, we mean the set of bodies $m$ coordinatizes at $\vx$:
\begin{equation}
\ev_m(\vx)\de\{ b : \W(m,b,\vx)\}.
\end{equation}

\begin{description}
\item[\underline{\ax{AxEv}}]
All inertial observers coordinatize the same set of events:
\begin{equation*}
\IOb(m)\land\IOb(k)\rightarrow  \exists \vy\, \forall b
\big[\W(m,b,\vx)\leftrightarrow\W(k,b,\vy)\big].
\end{equation*}
\end{description}
From now on, we will use $\ev_m(\vx)=\ev_k(\vy)$ to abbreviate the
subformula $\forall b [\W(m,b,\vx)\leftrightarrow\W(k,b,\vy)]$ of
\ax{AxEv}. 

These three axioms are enough to capture the essence of special
relativity. However, let us assume two more simplifying axioms.

\begin{description}
\item[\underline{\ax{AxSelf}}]
Any inertial observer is stationary relative to himself:
\begin{equation*}
\IOb(m)\rightarrow \forall \vx\big[\W(m,m,\vx) \leftrightarrow x_2=\ldots=x_d=0\big].
\end{equation*}
\end{description}
Our last axiom on inertial observers is a symmetry axiom saying that
they use the same units of measurement.

\begin{description}
\item[\underline{\ax{AxSymD}}]
Any two inertial observers agree as to the spatial distance between
two events if these two events are simultaneous for both of them;
furthermore, the speed of light is 1 for all observers:
\begin{multline*}
\IOb(m)\land\IOb(k) \land x_1=y_1\land x'_1=y'_1\land
\ev_m(\vx)=\ev_k(\vx') \\\land \ev_m(\vy)=\ev_k(\vy') \rightarrow
\sqspace(\vx,\vy)=\sqspace(\vx',\vy')
\text{ and }\\
\IOb(m)\rightarrow\exists
p\big[\Ph(p)\land\W(m,p,0,\ldots,0)\land\W(m,p,1,1,0,\ldots,0)\big].
\end{multline*}
\end{description}

Let us introduce an axiom system for special relativity as the collection
of the five simple axioms above:
\begin{equation*}\label{specrel}
\ax{SpecRel_d} \de \ax{AxPh}+\ax{AxOField}+ \ax{AxEv}+\ax{AxSelf}+
\ax{AxSymD}.
\end{equation*}

To show that the five simple axioms of \ax{SpecRel_d} capture special
relativity well, let us introduce the concept of {\bf worldview
  transformation} between observers $m$ and $k$ (in symbols,
$\w_{mk}$) as the binary relation on $\Q^d$ connecting the coordinate
points where $m$ and $k$ coordinatize the same events:
\begin{equation}\label{eq-ww}
\w_{mk}(\vx,\vx')\defiff \forall b\big[\W(m,b,\vx)\leftrightarrow \W(k,b,\vx')\big].
\end{equation}

Map $P:\Q^d\rightarrow\Q^d$ is called a {\bf Poincar\'e transformation} iff
it is an affine bijection having the following property
\begin{equation}
\timed(\vx,\vy)^2-\sqspace(\vx,\vy)=\timed(\vx',\vy')^2-\sqspace(\vx',\vy')
\end{equation}
for all $\vx,\vy,\vx',\vy'\in\Q^d$ for which  $P(\vx)=\vx'$ and $P(\vy)=\vy'$.

Theorem~\ref{thm-poi} shows that our axiom system \ax{SpecRel_d}
captures the kinematics of special relativity since it implies that
the worldview transformations between inertial observers are Poincar\'e transformations.
\begin{thm}\label{thm-poi}
Let $d\ge3$. Assume \ax{SpecRel_d}. Then $\w_{mk}$ is a Poincar\'e
transformation if $m$ and $k$ are inertial
observers.
\end{thm}
For the proof of Theorem~\ref{thm-poi}, see \cite{wnst}.  For a
similar result over Euclidean fields, see, e.g., \cite[Thms. 1.4 \&
  1.2]{AMNSamples}, \cite[Thm. 11.10]{logst}, \cite[Thm.3.1.4]{SzPhd}.

Let us now introduce a further auxiliary axiom about the possibility
of motion of inertial observers.
\begin{description}\label{axthexp}
\item[\underline{\ax{AxThExp}}] Inertial observers can move along any
  straight line with any speed less than the speed of light:
\begin{multline*}
\exists h \; \IOb(h)\land
\big(\IOb(m)\land \sqspace(\vx,\vy)<\timed(\vx,\vy)^2
  \\ \rightarrow \exists k \big[\IOb(k)\land \W(m,k,\vx)\land\W(m,k,\vy)\big]\big).
\end{multline*}
\end{description}

Theorem~\ref{thm-eof} below shows that axiom
\ax{AxThExp} implies that positive numbers have square roots if
\ax{SpecRel_d} is assumed.
\begin{thm}\label{thm-eof}If $d\ge3$, then
\begin{equation*} 
\num(\ax{SpecRel_d} + \ax{AxThExp})=\{\mathfrak{Q}:\mathfrak{Q}
\text{ is a Euclidean field} \}.
\end{equation*}
\end{thm}

\begin{rem}\label{rem-of}
Axiom \ax{AxThExp} cannot be omitted from Theorem~\ref{thm-eof} since
\ax{SpecRel_d} has a model over every ordered field, i.e., for all
$d\ge2$,
\begin{equation*}
\num(\ax{SpecRel_d})=\{\mathfrak{Q}:\mathfrak{Q}
\text{ is an ordered field} \} 
\end{equation*}
for all $d\ge2$.  Moreover, \ax{SpecRel_d} also has non trivial
 models in which there are several observers moving relative to each
 other. We conjecture that there is a model of \ax{SpecRel_d} over
 every ordered field such that the possible speeds of observers are
 dense in interval $[0,1]$, see Conjecture~\ref{conj-of} on
 p.\pageref{conj-of}.
\end{rem}

Since our measurements have only finite accuracy, it is natural to
assume \ax{AxThExp} only approximately. 
\begin{description}\label{axthexp-}
\item[\underline{\ax{AxThExp^-}}] Inertial observers can move roughly
  with any speed less than the speed of light roughly in any
  direction:
\begin{multline*}  
\exists h \; \IOb(h) \land \Big(\IOb(m)\land \varepsilon >0 \land
v_2^2+\ldots+v_d^2<1\land v_1=1 \\\rightarrow \exists \vw \Big[
  (w_1-v_1)^2+\ldots+(w_d-v_d)^2<\varepsilon \land \forall \vx\vy\, \exists
  \lambda\Big(\vx-\vy=\lambda \vw \\ \rightarrow \exists k\big[
    \IOb(m)\land \W(m,k,\vy)\land\W(m,k,\vy)\big]\Big)\Big]\Big).
\end{multline*}
\end{description}

By Theorem~\ref{thm-rac}, a model of \ax{SpecRel_d} + \ax{AxThExp^-}
has a model over the field of rational numbers in any dimension. We
use the notation $\mathfrak{Q}\in\num(\ax{Th})$ for algebraic
structure $\mathfrak{Q}$ the same way as the model theoretic notation
$\mathfrak{Q}\in Mod(\ax{AxField})$, e.g., $\rac\in\num(\ax{Th})$
means that $\rac$, the field of rational numbers, can be the structure
of quantities in theory \ax{Th}.

\begin{thm}\label{thm-rac} For all $d\ge2$,
\begin{equation*}
\rac\in\num(\ax{SpecRel_d} + \ax{AxThExp^-}).
\end{equation*}
\end{thm}
For the proof of Theorem~\ref{thm-rac}, see Section\ref{sec-proof}.

An ordered field is called {\bf Archimedean field} iff for all
$a$, there is a natural number $n$ such that
\begin{equation}
a<\underbrace{1+\ldots+1}_n
\end{equation}
holds.  By Pickert--Hion Theorem, every Archimedean field is
isomorphic to a subfield of the field of real numbers, see, e.g.,
\cite[\S VIII]{fuchs}, \cite[C.44.2]{CHA}. Consequently, the field of
rational numbers is dense in any Archimedean field since it is dense
in the field of real numbers. Therefore, the following is a corollary
of Theorem \ref{thm-rac}.
\begin{cor}\label{cor-arch} For all $d\ge2$,
\begin{equation*}
\{\mathfrak{Q}:\mathfrak{Q} \text{ is an Archimedean 
  field}\}\subsetneqq\num(\ax{SpecRel_d}\! +\!
\ax{AxThExp^-}).
\end{equation*}
\end{cor}

The question ``exactly which ordered fields can be the quantity
structures of theory \ax{SpecRel_d} + \ax{AxThExp^-}?'' is open. By
L\"ovenheim--Skolem Theorem it is clear that $\num(\ax{SpecRel_d}\!
+\! \ax{AxThExp^-})$ cannot be the class of Archimedean fields since
it has elements of arbitrarily large cardinality while an Archimedean
field has at most the cardinality of continuum since Archimedean
fields are subsets of the field of real numbers by Pickert--Hion
Theorem. We conjecture that there is a model of \ax{SpecRel_d} +
\ax{AxThExp^-} over every ordered field in any dimension, i.e.:
\begin{conj}\label{conj-of} For all $d\ge2$,
\begin{equation*}
\num(\ax{SpecRel_d} + \ax{AxThExp^-})=\{\mathfrak{Q}:\mathfrak{Q}
\text{ is an ordered field}\}.
\end{equation*}
\end{conj}

\section{Proof of Theorem~\ref{thm-rac}}
\label{sec-proof}
In this section, we are going to prove our main result. To do so, let
us recall some concepts and theorems from the literature. The
following theorem is well-known, see, e.g., \cite[Thm.2.1]{schmutz}.
\begin{thm} \label{thm-dc}
The unit sphere of $\R^n$ has a dense set of points with rational coordinates.
\end{thm}

The {\bf Euclidean length} of $\vx\in\Q^n$ if $n\ge 1$ is defined as:
\begin{equation}
|\vx|\de\sqrt{x_1^2+\dots+x_n^2}.
\end{equation}

Let us recall that the {\bf norm} of linear map
$A:\R^d\rightarrow\R^d$, in symbols $||A||$, is defined as follows:
\begin{equation}
||A||\de\max\{ |A\vx|: \vx\in\R^d \text{ and } |\vx|=1\}.
\end{equation}
Linear bijection $A$ is called {\bf orthogonal transformation} if it
preserves the Euclidean distance.

Theorem~\ref{thm-dc} implies Theorem~\ref{thm-dot}, see
\cite[Thm.3.1]{schmutz}.
\begin{thm} \label{thm-dot}
For all orthogonal transformation $T:\R^n\rightarrow\R^n$ and any
$\varepsilon>0$, there is a orthogonal transformaion
$A:\rac^n\rightarrow\rac^n$ such that $||T-A||<\varepsilon$.
\end{thm}

Using Theorem~\ref{thm-dot}, let us prove that its statement also
holds for Poincar\'e transformations.
\begin{thm}\label{thm-lor}
For every Poincar\'e transformation $L:\R^d\rightarrow\R^d$ and
positive real number $\varepsilon$, there is a Poincar\'e
transformation $L^*:\rac^d\rightarrow\rac^d$ such that
$||L-L^*||<\varepsilon$.
\end{thm}

We are going to prove Theorem~\ref{thm-lor} by using the fact that
every Poincar\'e transformation is a composition of a Lorentz boost
and two orthogonal transformations. {\bf Lorentz boost} corresponding
to velocity $v\in[0,1)$, in symbols $B_v$, is defined as the following
  linear map:
\begin{equation}
B_v\vx=\left\langle
\frac{x_1-vx_2}{\sqrt{1-v^2}},\frac{x_2-vx_1}{\sqrt{1-v^2}},x_3,\ldots,x_d\right\rangle\quad\text{for
  all } \vx\in\Q^d.
\end{equation}

\begin{lem}\label{lem-b}
For all Lorentz boost $B_v:\R^d\rightarrow\R^d$ and positive number
$\varepsilon$, there is a Lorentz boost $B_w:\rac^d\rightarrow\rac^d$
such that $||B_v-B_w||<\varepsilon$.
\end{lem}

\begin{proof}
Since, by Theorem~\ref{thm-dc}, the set of rational points are dense
in the unit circle, we have that, for all $\delta>0$ and $v\in[0,1)$,
  there is a $w\in \rac\cup [0,1)$ such that $|v-w|<\delta$ and
    $\sqrt{1-w^2}\in\rac$, i.e., $B_w$ takes rational point to
    rational ones. So we have to show that $||B_v- B_w||<\varepsilon$
    if $\delta$ is small enough. Since in a finite-dimensional vector
    space all norms are equivalent, see \cite[\S 8.5]{BL}, it is
    enough to show that the norm of $B_v-B_w$ can be less
    than any positive real number according to the Euclidean norm,
    which is

\begin{equation}
\sqrt{2\left|\frac{1}{\sqrt{1-v^2}} -\frac{1}{\sqrt{1-w^2}}\right|^2 
+2\left|\frac{v}{\sqrt{1-v^2}} -\frac{w}{\sqrt{1-w^2}}\right|^2}.
\end{equation}
By the continuity of functions $v\mapsto (1-v^2)^{-\frac{1}{2}}$ and
$v\mapsto v(1-v^2)^{-\frac{1}{2}}$, the Euclidean norm of $B_v-B_w$ is less than any fixed positive real number if $|v-w|$ is
small enough. Therefore, there is a Lorentz boost $B_w$ such that $B_w$
maps rational points to rational ones and $||B_w-B_v||<\varepsilon$. 
\end{proof}

\begin{lem}\label{lem-comp}
Let $A$ and $B$ be linear bijections of $\R^d$. Let $A'$ and $B'$
linear maps such that $||A-A'||<\varepsilon_1$ and
$||B-B'||<\varepsilon_2$. Then $||BA-B'A'||\le
\varepsilon_1||B||+\varepsilon_1\varepsilon_2+ \varepsilon_2||A||$.
\end{lem}

\begin{proof}
First let us note that
\begin{equation}
||A'||=||A'-A+A||\le ||A'-A||+||A||=\varepsilon_1 + ||A||
\end{equation}
by the triangle inequality. 
Let $\vx\in\R^d$ such that $|\vx|=1$. We have to show that 
\begin{equation}
|BA\vx-B'A'\vx|\le\varepsilon_1||B||
 +\varepsilon_1\varepsilon_2 +\varepsilon_2||A||
\end{equation}
By the triangle inequality and the fact that
$|M\vy|\le||M||\cdot|\vy|$, we have
\begin{multline}
|BA\vx-B'A'\vx|=|BA\vx-BA'\vx+BA'\vx-B'A'\vx|\\\le
|BA\vx-BA'\vx|+|BA'\vx-B'A'\vx|\\ 
\le||B||\cdot|A\vx-A'\vx|+||B-B'||\cdot|A'\vx| 
\\\le||B||\cdot||A-A'||+||B-B'||\cdot||A'|| \\\le\varepsilon_1||B|| +
\varepsilon_2(\varepsilon_1 +||A||) =\varepsilon_1||B|| 
+\varepsilon_1\varepsilon_2+ \varepsilon_2||A||,
\end{multline}
and this is what we wanted to prove.
\end{proof}

\begin{proof}[Proof of Theorem~\ref{thm-lor}]
Every Poincar\'e transformation is a composition of a translation, a Lorentz-boost $B_v$
and an orthogonal transformation. Therefore, Lemmas~\ref{lem-b} and \ref{lem-comp},
together with Theorem~\ref{thm-dot} imply our statement.
\end{proof}

Now we are going to prove Theorem~\ref{thm-rac}. Let $\Id$ be the {\bf
  identity map} of $\rac^d$.  We denote the {\bf origin} of $\Q^n$ by
$\vo$, i.e., 
\begin{equation} 
\vo\de \langle 0,\ldots,0\rangle.
\end{equation}
Let the {\bf time-axis} be defined as the following subset of $\Q^d$:
\begin{equation}\label{eq-taxis}
\taxis\de\{\vx:x_2=\ldots=x_d=0\}.
\end{equation}
Let $H$ be a subset of $\Q^d$ and let $f:\Q^d\rightarrow\Q^d$ be a
map. The {\bf $f$-image} of set $H$ is defined as:
\begin{equation}
f[H]\de\{f(\vx):\vx\in H\}.
\end{equation}
The so-called {\bf worldline} of body $b$ according to observer $m$ is
defined as follows:
\begin{equation}
\wl_m(b)\de\{ \vx: \W(m,b,\vx)\}.
\end{equation}

\begin{proof}

\begin{figure}
\begin{tikzpicture}[scale=2]
\begin{scope}[xshift=-1.5cm]
\draw[blue, very thick] (0,-1)--(0,1) node[right]{$\Id$};
\draw[black!60!green, very thick] (-0.5,-1)--(0.5,1) node[right]{$m$};
\draw[red] (-1,-1)--(1,1);
\draw[red] (-1,1) --(1,-1);
\draw[red] (0,1) ellipse (1 and 0.1);
\draw[red] (0,-1) ellipse (1 and 0.1);
\draw[black!30!brown,thick] (.49,-1) -- node[right]{$b$} (0.26, 1);
\fill (0.3,0.6) node[right]{$m(\vx)$} circle (0.05);
\fill (0,0) circle (0.05);
\end{scope}
\node[above] at (0,0.5) {$m$};
\draw[->,>=stealth] (0.8,0.2) .. controls (0,0.5) .. (-0.8,0.2);
\begin{scope}[xshift=1.5cm]
\draw[black!60!green, very thick] (0,-1)--(0,1) node[right]{$m$};
\draw[black!30!brown, thick] (-0.25,1) -- node[right]{$b$} (0.75, -1);
\draw[red] (-1,-1)--(1,1);
\draw[red] (-1,1)--(1,-1);
\draw[red] (0,1) ellipse (1 and 0.1);
\draw[red] (0,-1) ellipse (1 and 0.1);
\fill (0,0.5) node[right]{$\vx$} circle (0.05);
\fill (0,0) node[left]{$\vo$}  circle (0.05);
\end{scope}
\end{tikzpicture}
\caption{\label{fig-w} Illustration for the proof of Theorem~\ref{thm-rac}}
\end{figure}
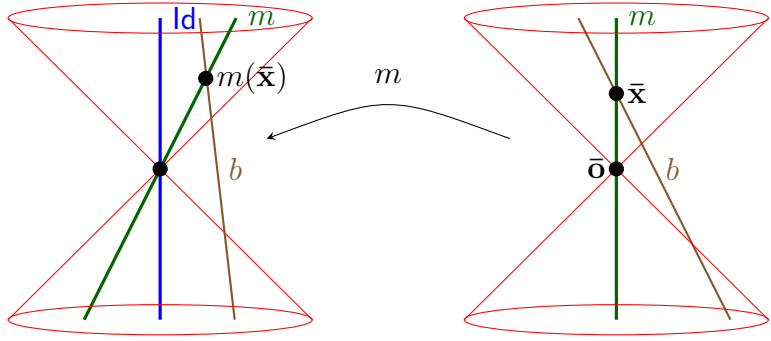

We are going to construct a model of \ax{SpecRel_d} + \ax{AxThExp^-} over
$\rac$.  So let $\langle \Q, +,\cdot,\le\rangle$ be the ordered field
of rational numbers. Let 
\begin{equation}
\Ph \de \{l : l \text{ is a line of slope 1}\},
\end{equation}
\begin{equation}\label{eq-IOb}
\IOb\de\left\{m: m \text{ is a Poincar\'e transformation from $\rac^d$ to $\rac^d$}\right\},
\end{equation}
and let $\B=\IOb\cup\Ph$. 
First we are going to give the worldview of observer
$\Id$. Let 
\begin{equation}
\W(\Id,\Id,\vx)\defiff x_2=\ldots=x_d=0;
\end{equation} 
for any other inertial observer $m$, let
\begin{equation}\label{eq-Idm}
\W(\Id,m,\vx) \defiff \vx\in m[{\taxis}]; 
\end{equation}
and for any light signal $p\in\Ph$, let
\begin{equation}\label{eq-Idb}
\W(\Id,p,\vx)\defiff \vx\in p.
\end{equation}
Now the worldview of observer $\Id$ is given.  From the worldview of
$\Id$, we construct the worldview of another inertial observer $m$ as
follows:
\begin{equation}\label{ww-def}
\W(m,b,\vx)\defiff \W\big(\Id,b,m(\vx)\big)
\end{equation}
for all body $b\in\B$, see Figure~\ref{fig-w}.  

Now we have given the model. Let us see why the axioms of
\ax{SpecRel_d} and \ax{AxThExp^-} are valid in it.

By the above definition of $\W$, if $m$ and $k$ are inertial
observers, then
\begin{equation}\label{eq-iob}
\W(m,k,\vx)\enskip \text{ holds iff }\enskip m(\vx)\in k[\taxis],
\end{equation}
and if $m\in\IOb$ and $p\in\Ph$, then
\begin{equation}\label{eq-b}
\W(m,p,\vx)\enskip\text{ holds iff }\enskip
m(\vx)\in p.
\end{equation}
The worldview transformations between inertial
observers $m$ and $\Id$ is $m$, i.e., $\w_{m\Id}=m$ by equation
\eqref{ww-def}. Therefore, the worldview transformation between
inertial observers $m$ and $k$ is $k^{-1}\circ m$, i.e., 
\begin{equation}\label{eq-w}
\w_{mk}=k^{-1} \circ m
\end{equation}
since $\w_{mk}=\w_{\Id k}\circ\w_{m\Id}$ and $\w_{\Id k}=(\w_{k\Id
})^{-1}$ by the definition of the worldview
transformation \eqref{eq-ww}. Specially, the worldview transformations between
inertial observers are Poincar\'e transformations in these models (as
Theorem~\ref{thm-poi} requires it). Hence
\begin{equation}\label{eq-bij}
\w_{mk} \text{ is a bijection for all inertial observers $m$ and $k$.}
\end{equation}

Axiom \ax{AxPh} is valid for observer $\Id$ by the definition of $\Ph$
and that of his worldview. It is also clear that the speed of light is
1 for observer $\Id$. Axiom \ax{AxPh} is valid for the other observers
since Poincar\'e transformations take lines of slope one to lines of
slope one.  This also show that the speed of light is $1$ according to
every inertial observer, which is the second half of \ax{AxSymD}.

Axiom \ax{AxOField} is valid in this model since
$\rac$ is an ordered field.
Axiom \ax{AxEv} is valid in this model since Poincar\'e
transformations are bijections.
Axiom \ax{AxSelf} is valid in this model since
\begin{multline}
\W(m,m,\vx)\stackrel{\eqref{eq-iob}}{\iff} m(\vx)\in
m[\taxis]\\\stackrel{\eqref{eq-IOb}}{\iff} \vx\in\taxis
\stackrel{\eqref{eq-taxis}}{\iff} x_2=\ldots=x_d=0.
\end{multline}

Any Poincar{\'e} transformation $P$ preserves the spatial distance of
points $\vx,\vy$ for which $x_1=y_1$ and
$P(\vx)_1=P(\vy)_1$. Therefore, inertial observers agree as to the
spatial distance between two events if these two events are
simultaneous for both of them. We have already shown that the speed of
light is 1 for each inertial observers in this model. Hence axiom
\ax{AxSymD} is also valid in this model.

Now we are going to show that \ax{AxThExp^-} is valid in this model.
The $\exists h\, \IOb(h)$ part of axiom \ax{AxThExp} is valid, since
there are Poincar\'e transformations (e.g., $\Id$ is one).  To show
that the rest of axiom \ax{AxThExp^-} is valid, let $m$ be an inertial
observer and let us fix an $\varepsilon>0$ and a $\vv\in\rac^d$ for
which $v_2^2+\ldots v_d^2<1$ and $v_1=1$.  Let $\vet$ be vector
$\langle1,0,\ldots,0\rangle$. Let $L$ be a Lorentz transformation
(i.e., linear Poincar\'e transformation) for which
\begin{equation}\label{v}
\vv=\frac{L(\vet)}{L(\vet)_1}.
\end{equation}
Let $0<\delta<1$ be such that 
\begin{equation}\label{A}
\delta < \frac{\varepsilon L(\vet)_1}{2}\quad\text{ and }
\end{equation}
\begin{equation}\label{B}
\left|\frac{1}{L(\vet)_1}-\frac{1}{x}\right|<\frac{\varepsilon}{2(||L||+1)}
\end{equation}
for any $x$ for which $|L(\vet)_1-x|<\delta$.  By
Theorem~\ref{thm-lor}, there is a Lorentz transformation $L^*$ which
takes rational points to rational ones and $||L-L^*||<\delta$.  Then
\begin{equation}\label{C}
|L(\vet)-L^*(\vet)|<\delta
\end{equation}
since $|\vet|=1$.  We have $|L(\vet)_1-L^*(\vet)_1|<\delta$ since
$|x_1|<|\vx|$ for all $\vx\in\rac^d$. By triangle inequality, we also
have
\begin{equation}\label{E}
|L^*(\vet)|\le|L^*(\vet)-L(\vet)|+|L(\vet)|\le
\delta+||L||\le
1+||L||. 
\end{equation}
Let 
\begin{equation}\label{w}
\vw\de \frac{L^*(\vet)}{L^*(\vet)_1}
\end{equation}
By triangle inequality, we have
\begin{multline}
|\vv-\vw|\stackrel[\eqref{v}]{\eqref{w}}{=}\left|\frac{L(\vet)}{L(\vet)_1}
-\frac{L^*(\vet)}{L^*(\vet)_1}\right|\\ =\left|\frac{L(\vet)}{L(\vet)_1}
-\frac{L^*(\vet)}{L(\vet)_1} +\frac{L^*(\vet)}{L(\vet)_1} -
\frac{L^*(\vet)}{L^*(\vet)_1}\right|\\ \le
\left|\frac{L(\vet)}{L(\vet)_1} -\frac{L^*(\vet)}{L(\vet)_1}\right|
+|L^*(\vet)|\left|\frac{1}{L(\vet)_1}
-\frac{1}{L^*(\vet)_1}\right|\\ \stackrel[\eqref{E}]{}{\le} 
\frac{1}{L(\vet)_1}\left|L(\vet)
-L^*(\vet)\right| +(1+||L||)\left|\frac{1}{L(\vet)_1}
-\frac{1}{L^*(\vet)_1}\right|\\ \stackrel[\eqref{B}]{\eqref{C}}{<}\frac{\delta}{L(\vet)_1} +
\frac{\varepsilon}{2}\stackrel[\eqref{A}]{}{<}\varepsilon
\end{multline}
  Let $\vx,\vy\in\rac^d$ such that there is a $\lambda\in\rac$ such
  that $\vy-\vx=\lambda \vw$. To finish the proof of \ax{AxThExp^-},
  we have to show that there is an inertial observer $k$ such that
  $\W(m,k,\vx )$ and $\W(m,k,\vy)$, i.e., $\vx,\vy\in\wl_m(k)$. Let
  $P^*=L^*+\vx$. $P^*$ is a Poincar\'e transformation taking rational
  points to rational ones. Therefore, there is an inertial observer
  $k$ such that $\w_{km}=P^*$. Since $\wl_m(k)=\w_{km}[\taxis]$, we
  have that $\w_{km}(\vo)=\vx\in\wl_m(k)$ and that $\vy=a
  L^*(\vet)+\vx=\w_{km}(a\vet)\in\wl_m(k)$, where
  $a=\lambda/L^*(\vet)_1$. This shows that \ax{AxThExp^-} is also
  valid in our model.
\end{proof}

\section{Acknowledgments}
This research is supported by the Hungarian Scientific Research Fund
for basic research grants No.~T81188 and No.~PD84093, as well as by a
Bolyai grant for J.~X.~Madar\'asz.

\bibliography{LogRelBib}
\bibliographystyle{plain}

\end{document}